\documentclass[twocolumn,amsmath,amssymb,letterpaper,nofootinbib]{revtex4}

\usepackage{enumerate}

\usepackage{graphicx}

\usepackage{amsmath, amssymb, amsthm, amsfonts}

\newcommand{\size}[1]{\left|#1\right|} 

\newcommand{\pr}{^\prime}

\newcommand{\sym}[1]{\text{Sym}(#1)}

\newcommand{\be}{\begin{equation}}
\newcommand{\ee}{\end{equation}}

\newcommand{\ket}[1]{|#1\rangle}
\newcommand{\bra}[1]{\langle #1|}
\newcommand{\ketbra}[2]{\ket{#1}\!\bra{#2}}        

\newcommand{\extwk}{\widehat{W}^{(k)}}

\newcommand{\R}{\mathcal{R}}
\newcommand{\Z}{\mathcal{Z}}

\newcommand{\cel}{\text{Cel}}
\newcommand{\matv}{\text{Mat}_V}
\newcommand{\tr}{\text{tr}}

\newtheorem{definition}{Definition}
\newtheorem{prop}{Proposition}
\newtheorem{theorem}{Theorem}
\newtheorem{corollary}{Corollary}
\newtheorem{lemma}{Lemma}

\begin{document}

\title{$k$-Boson Quantum Walks Do Not Distinguish Arbitrary Graphs}
\author{Jamie Smith\footnote[2]{Supported by the Natural Sciences and Engineering Research Council of Canada}}
\affiliation{Institute for Quantum Computing and Department of Combinatorics \&
Optimization, \\University of Waterloo, Waterloo N2L 3G1, ON Canada}
\email{ja5smith@iqc.ca}

\begin{abstract}In this paper, we define $k$-equivalence, a relation on graphs that relies on their associated cellular algebras.  We show that a $k$-Boson quantum walk cannot distinguish pairs of graphs that are $k$-equivalent.  The existence of pairs of $k$-equivalent graphs has been shown by Ponomarenko et al. \cite{Barghi:2009p42, Evdokimov:1999p668}.  This gives a negative answer to a question posed by Gamble et al. \cite{Gamble:2010p1085}.
\end{abstract}

\maketitle

\section{Introduction}

In a recent paper, Gamble, Friesen, Zhou, Joynt and Coppersmith \cite{Gamble:2010p1085} consider the application of a multi-paritcle quantum walks to the graph isomorphism problem.   In particular, they examine (a) two non-interacting Bosons; (b) two non-interacting Fermions; and (c) two interacting Bosons.  In each case, a graph $G=(V,E)$ is associated with an Hamiltonian $H$; this Hamiltonian will depend on the dynamics and number of particles chosen.  The corresponding unitary operator is then given by
$$U=e^{-itH}$$
A detailed description of continuous time quantum walks and their properties can be found in \cite{Childs:2008p1088}.

The {\bf Green's functions} associated with this graph are the values $$\mathcal{G}(i,j)=\bra{j}U\ket{i}$$ where $i$ and $j$ run over an appropriate basis (in the case of \cite{Gamble:2010p1085}, this is the two-particle basis).  We say that two graphs are {\bf distinguishable} if their sets of Green's functions differ.

Gamble et al. prove that non-interacting Bosons and Fermions fail to distinguish pairs of non-isomorphic strongly regular graphs with the same parameters (see Definition \ref{srg}).  In order to evaluate the effectiveness of interacting Bosons, they consider all tabulated pairs of non-isomorphic strongly regular graphs with up to 64 vertices.  They find that the interacting 2-Boson walk does indeed distinguish all the pairs of strongly regular graphs they considered.  They go on to ask if, for a high enough value of $k$, an interacting $k$-Boson quantum walk could distinguish all pairs of non-isomorphic graphs.  This would place the graph isomorphism problem in P.  In this paper, we show that for any $k$, there are pairs of non-isomorphic graphs that are not distinguished by an interacting $k$-Boson quantum walk.

In order to prove this, we introduce cellular algebras in Section \ref{CA}.  We define weak and strong notions of isomorphisms between cellular algebras, as well as the construction of these algebras from graphs.  Section \ref{kequiv} defines the $k$-equivalence of graphs, which is based on the properties of their associated cellular algebras.  Sections \ref{unit}-\ref{kboson} apply these ideas to multi-particle walks on graphs, showing that two $k$-equivalent graphs are not distinguished by a $k$-Boson quantum walk.  Section \ref{2boson} is concerned with the case of a two-Boson quantum walk, directly addressing \cite{Gamble:2010p1085}.  Section \ref{kboson} generalizes this to $k$ Bosons.  The existence of $k$-equivalent graphs for any positive integer $k$ is proven in \cite{Barghi:2009p42} and \cite{Evdokimov:1999p668}.  In fact, these are the same graphs that are used to demonstrate that the spectrum of the $k$-symmetric power of a graph do not distinguish arbitrary graphs (see \cite{Audenaert:2005p14,Barghi:2009p42}).  Indeed, many of the results in this paper are analogous to those of \cite{Barghi:2009p42}; the main innovation in this paper is applying these results to multi-particle quantum walks.

\section{Cellular Algebras}\label{CA}

\subsection{Definitions}

Cellular algebras are a generalization of coherent configurations, which were developed by Weisfeiler and Lehman (\cite{Weisfeiler:1968fk}) and Higman (\cite{HIGMAN:1970p409}) as an approach to the graph isomorphism problem.  We will see how these algebras are generated from graphs;  these generated cellular algebras capture structural information about the underlying graph that we will use to prove the main theorem of this paper.

Let $V$ be a finite vertex set.  Then, $\matv$ is the algebra of all complex-valued matrices indexed by $V$.  We define a cellular algebra as follows:

\begin{definition} Let $W$ be a subalgebra of $\matv$.  Then, $W$ is a cellular algebra if the following hold:
\begin{enumerate}[(i)]
\item $W$ is closed under Hadamard multiplication $\circ$.
\item $W$ is closed under complex conjugation $^\dagger$.
\item $W$ contains the identity $I$ and the all-ones matrix $J$.
\end{enumerate}
\end{definition}

The following is a useful consequence of this definition:

\begin{prop} If $W$ is a cellular algebra, then
\begin{enumerate}[(i)]
\item $W$ has a unique basis of $0-1$ matrices $\R$, and $\sum_{R\in\R}R=J$.
\item There is a subset $C\subseteq\R$ such that $\sum_{R\in\R}R=I$.
\item If $R\in\R$, then $R^\dagger\in\R$
\end{enumerate}
\end{prop}

We call $\R$ the set of {\bf basis relations}\footnote{The term {\em basis relation} comes from an alternative  definition, in which we consider binary relations on $V$, rather than the equivalent $0-1$ matrices} of $W$.  We will also use $\R^*$ to denote the set of all sums of elements of $\R$; this is the set of {\bf relations} of $W$.  A set of vertices $U\subseteq V$ is called a {\bf cell} of $W$ if $I_U$, the identity on $U$, is a basis relation of $W$.  The set of cells of $W$ is denoted by $\cel(W)$.

\subsection{The Cellular Closure: Cellular Algebras from Graphs}\label{secsrg}

The cellular algebra $W=[M_1,...,M_\ell]$ is the smallest cellular algebra containing $\{M_1,...,M_\ell\}$, a set of $n\times n$ matrices.  We say that $W$ is {\bf generated} by $\{M_1,...,M_\ell\}$.  If $G=(V,E)$ is a graph with adjacency matrix $A$, then we say that $W_G=[A]$ is the {\bf cellular closure} of $G$.  This will sometimes be denoted by $[G]$.  Strongly regular graphs have the simplest cellular closures:

\begin{definition}\label{srg}
A {\bf strongly regular graph} $G$ is associated with a set of parameters $(n,k,\lambda,\mu)$ such that:
\begin{enumerate}[(i)]
\item $G$ has $n$ vertices.
\item Each vertex has degree $k$.
\item Each pair of adjacent vertices share $\lambda$ common neighbours.
\item Each pair of non-adjacent vertices share $\mu$ common neighbours.
\end{enumerate}
\end{definition}

If a strongly regular graph $G$ has adjacency matrix $A$, it is easily verified that $\{I, A, (J-I-A)\}$ form the basis for the cellular algebra $W_G=[A]$.  While most generated cellular algebras are not as straightforward as this, the Weisfeiler-Lehman algorithm (see \cite{Weisfeiler:1968fk}, \cite{Evdokimov:1999p1005}) calculates the cellular closure of a set of matrices in polynomial time.

\subsection{Weak and Strong Isomorphisms}

We will define two notions of isomorphisms between cellular algebras--- one of a combinatorial nature (strong), and the other of an algebraic nature (weak).  Let $W$ and $W\pr$ be cellular algebras with vertex sets $V$ and $V\pr$ and basis relations $\R$ and $\R\pr$, respectively. 

\begin{definition} A {\bf weak isomorphism} is a bijection $\phi\colon W\to W\pr$ that preserves addition, matrix multiplication, Hadamard multiplication and complex conjugation.
\end{definition}
Two immediate consequences of this definition are that $\phi(I)=I$, and $\phi$ is a bijection from the basis relations $\R$ to $\R\pr$.  We also note that $\phi$ induces a bijection between cells, $\phi\pr\colon \cel(W)\to\cel(W\pr)$.  We will use the following lemma and its corollary in the proof of Theorem \ref{indistinguishable}.

\begin{lemma}\label{celltocell}
Take $X\in\cel(W)$.  Then $\size{X}=\size{\phi\pr(X)}$.\end{lemma}
\begin{proof} Let $\R_X$ be the set basis relations of $W$ restricted to the vertex set $X$.  Define $\R\pr_{\phi\pr(X)}$ analogously.  We will first show that $J_X$, the all-ones matrix on the set $X$, is mapped to $J_{\phi\pr(X)}$, the all-ones matrix on the corresponding cell in $W\pr$.  Since, $\phi(I_X)=I_{\phi\pr(X)}$ and for any $R\in\R_X$,
$$I_X\cdot R\cdot I_X=R$$
it follows that 
$$I_{\phi\pr(X)}\cdot\phi(R)\cdot I_{\phi\pr(X)}=\phi(R).$$

Therefore, if $R\in\R_X$, then $\phi(R)\in\R\pr_{\phi\pr(X)}$.  Moreover,

$$\phi(J_X)=\sum_{R\in\R_X}\phi(R)=\sum_{R\in\R\pr_{\phi\pr(X)}}R=J_{\phi\pr(X)}$$

Now, comparing
$$J_X\cdot J_X=\size{X}\cdot J_X$$
and
$$J_{\phi\pr(X)}\cdot J_{\phi\pr(X)}=\size{\phi\pr(X)}\cdot J_{\phi\pr(X)}$$
gives us $\size{X}=\size{\phi\pr(X)}$.

\end{proof}

\begin{corollary}\label{treq}
For all $R\in W$, $\tr(R)=\tr(\phi(R))$.
\end{corollary}
\begin{proof} Let $C$ and $C\pr$ be the basis relations of $W$ and $W\pr$ that sum to the identity.  Then, for each $R\in W$ and $I_X\in C$,
$$R\circ I_X=q_R(X)\cdot I_X$$
for some $q_R(X)\in \mathbb{C}$ and
$$\tr(R)=\sum_{X\in\cel(W)}q_R(X)\cdot\size{X}$$
Since $\phi$ is a weak isomorphism, and applying Lemma \ref{celltocell},
\begin{flalign*}
\tr(\phi(R))=&\sum_{X\in\cel(W)}q_R(X)\cdot\size{\phi\pr(X)}\\
=&\tr(R)
\end{flalign*}
\end{proof}

Weak isomorphisms respect algebraic structure, but do not take into account the vertices underlying the cellular algebra.

\begin{definition} A {\bf strong isomorphism} is a bijection $\psi\colon V\to V\pr$ such that, for each $R\in W$, there is a unique $R\pr\in W\pr$ such that
$$\forall u,v\in V,\: R(u,v)=R\pr(\psi(u),\psi(v)).$$
\end{definition}
Note that a strong isomorphism $\psi$ induces a weak isomorphism.  Unfortunately, not all weak isomorphisms are induced by a strong isomorphism.

\subsection{Cellular Algebra Extensions}
The $k$-extension of a cellular algebra $W$ is a larger algebra that contains additional structural information about $W$.  Before constructing the $k$-extension, we first need to define the centralizer algebra.
\begin{definition}
Let $G$ be a group acting on the set $S$.  The {\bf centralizer algebra} is defined as follows:
$$\Z(G,S)=\{A\in\text{Mat}_S\colon \forall g\in G,\;A^g=A\}$$
\end{definition}
In the next definition, we use the centralizer algebra $\Z(\sym{V},V^k)$.  In this case, $\sym(V)$ acts entrywise on $V^k$.
\begin{definition}
The {\bf $k$-extension} $\extwk$ of a cellular algebra $W$ is the smallest cellular algebra containing $W^k$ and $\Z(\sym{V},V^k)$:
$$\extwk=[W^k,\Z(\sym{V},V^k)]$$
\end{definition}
The following lemma is reproduced from \cite{Barghi:2009p42}, and was originally proven in \cite{Evdokimov:2000p1036}:
\begin{lemma}\label{cylindric1} Let $S=\{R_{i,j}\colon 1\leq i,j\leq k\}\subseteq \R^*$ be a set of relations.  Define the cylindric relation $\text{Cyl}_S$ such that, given $\overline{x},\overline{y}\in V^k$,
$$\text{Cyl}_S(\overline{x},\overline{y})=\prod_{i,j}R_{i,j}(x_i,y_j).$$
Then, $\text{Cyl}_S\in\extwk$.
\end{lemma}

\section{The $k$-Equivalence of Graphs} \label{kequiv}

We would like to use these ideas of weak isomorphism and $k$-extension to draw meaningful connections between graphs.

\begin{definition} Let $G$ and $G\pr$ be graphs with adjacency matrices $A$ and $A\pr$.  Then, $G$ and $G\pr$ are {\bf equivalent} if there is a weak isomorphism $\phi\colon [G]\to[G\pr]$ such that $\phi(A)=A\pr$.  We say that $\phi$ is a similarity from $G$ to $G\pr$.\end{definition}

Since $[G]$ is the smallest cellular algebra containing $A$, the weak isomorphism $\phi$ is in fact uniquely determined.  We now broaden this definition to take into account $k$-extensions

\begin{definition} Let $G$ and $G\pr$ be graphs with adjacency matrices $A$ and $A\pr$.  Then $\phi$ is a $k$-equivalence if
\begin{enumerate}[(i)]
\item It is an equivalence from $G$ to $G\pr$.
\item There exists a weak isomorphism $\widehat\phi\colon \widehat{W_G}^{(k)}\to \widehat{W_{G\pr}}^{(k)}$ such that
$$\widehat\phi\mid_{W_G^k}=\phi^k\qquad and \qquad\widehat\phi\mid_{\Z(\sym{V},V^k)}=I$$
\end{enumerate}
\end{definition}

Clearly, 1-equivalence corresponds with our existing definition of graph equivalence.  The following lemma appears in \cite{Barghi:2009p42}, and tells us that a $k$-equivalence acts in a convenient way on cylindric relations.

\begin{lemma}\label{cylindric2} Let $\phi$ be a $k$-equivalence from $G$ to $G\pr$.  Let $S\subseteq \R^*$ be a set of relations of $W_G$.  Then,
$$\widehat\phi(\text{Cyl}_S)=\text{Cyl}_{\phi(S)}$$
where $\phi(S)=\{\phi(R_{i,j})\colon R_{i,j}\in S\}$.
\end{lemma}

\section{Unitary Evolution and Cellular Algebras}\label{unit}

Let $H$ be a Hamiltonian and $W=(V, \R)$ a cellular algebra containing $H$.  Then, unitary $U$ corresponding to $H$ can be written as 

$$U=e^{-itH}=\sum_{j=0}^\infty\frac{(-itH)^j}{j!}$$

We will now take advantage of the fact that $U$ lies within $W$ to express the values of the Green's function in a convenient way.  Since $H\in W$, we can write each $H^j/j!$ as a linear combination of basis relations:

$$\frac{H^j}{j!}=\sum_{R\in\R}p_R(j)\cdot R$$

This gives us

$$U=e^{-itH}=\sum_{R\in\R}\left[R\cdot\sum_{j=0}^\infty p_R(j)\cdot(-it)^j\right]$$

This gives us a convenient way of expressing Green's function in terms of the basis relations $\R$.  Since each $R\in\R$ is a 0-1 matrix, the values of the Green function are given by
$$x_R(t)=\sum_{j=0}^\infty p_R(j)\cdot(-it)^j$$
each with multiplicity 
$$m_R=\text{sum}(R)=\text{tr}(RR^T)$$

\begin{theorem} \label{indistinguishable}Let $H$ and $H\pr$ be two Hamiltonians contained in cellular algebras $W$ and $W\pr$ respectively.  Furthermore, let $\phi:W\to W\pr$ be a weak isomorphism such that $\phi(H)=H\pr$. Then, the Green functions for $H$ and $H\pr$ take on the same values with the same multiplicities.\end{theorem}
\begin{proof}    Define $p\pr_R(j)$, $x\pr_R(t)$ and $m\pr_R$ as above.  Since $\phi$ is a weak isomorphism,
$$p_R(j)=p_{\phi(R)}\pr(j)$$
and therefore
$$x_R(t)=x\pr_{\phi(R)}(t)$$
Corollary \ref{treq} tells us that
$$m_R=\text{tr}(RR^T)=\text{tr}(\phi(RR^T))=m\pr_{\phi(R)}$$

Therefore, the Green functions for take on the same values with the same multiplicities.
\end{proof}

\section{Interacting 2-Boson Walks and 2-Equivalence}
\label{2boson}

We will first consider the case of two interacting Bosons.  We will extend this to to $k$ particles in the following section, but we include the 2-Boson case separately as it directly addresses \cite{Gamble:2010p1085}.  In \cite{Gamble:2010p1085}, the Hamiltonian for a two-Boson quantum walk is given by:

$$H_{2B}=-\frac{1}{2}(I+S)A^{\oplus2}+UR$$

where $U$ is a constant energy cost and

\begin{flalign*}
S=&\sum_{i,j}\ketbra{ij}{ji},\\
R=&\sum_i\ketbra{ii}{ii},\\
A^{\oplus n}=&(\overbrace{A\otimes I\otimes...\otimes I}^n)+...+(I\otimes I\otimes...\otimes A)
\end{flalign*}

The following lemma is a direct consequence of the definition of 2-extension, as well as Lemma \ref{cylindric2} regarding cylindric relations:
\begin{lemma} \label{HtoH} If $\phi$ is a 2-equivalence from $G$ to $G\pr$ with corresponding 2-Boson Hamiltonians $H_{2B}$ and $H\pr_{2B}$, then
\begin{enumerate}[(i)]
\item $\widehat\phi(S)=S$
\item $\widehat\phi(A^{\oplus2})=(A\pr)^{\oplus2}$
\item$\widehat\phi(R)=R$
\end{enumerate}
and therefore $\widehat{\phi}(H_{2B})=H\pr_{2B}$.

\end{lemma}
\begin{proof} See the proof of Lemma \ref{HtoHk} for a more general proof.\end{proof}

Combining Theorem \ref{indistinguishable} and Lemma \ref{HtoH}, we arrive at the following:

\begin{theorem} \label{2bos} If $G$ and $G\pr$ are 2-equivalent graphs, then they are not distinguished by the interacting 2-Boson walk.\end{theorem}

\section{Interacting $k$-Boson Walks and $k$-Equivalence}
\label{kboson}

We will now consider the $k$-boson case.  We will consider each term of the Hamiltonian in turn. First, we replace the term $(I+S)$ from the 2-particle Hamiltonian with $\sum_{\sym{k}}S$.  That is, we will work within the subspace invariant under any permutation of the $k$ particles.  Next we replace the term $A^{\oplus2}$ with $A^{\oplus k}$.  

Finally, we consider the interaction term.  We would like the energy contribution from each site to be a function of the number of particles at that site.  Let $\overline{x}\in V^k$ be a basis state of the $k$-Boson system.  Then, define $v_x$ as the number of particles at vertex $v$ in state $x$.  Then, $V_x=\{v_x:\;v\in V\}$.  Then, we can partition $V^k$ into equivalence classes $X_1,...,X_\ell$ such that $x$ and $y$ are in the same class if and only if $V_x=V_y$.  Let $R_i=\sum_{x\in X_i}\ketbra{x}{x}$.  Finally, to each of the $X_i$, we assign an energy penalty $U_i$.  This gives us the interaction term, $\sum_{i=1}^\ell U_iR_i$

Putting these together, we arrive at our $k$-Boson Hamiltonian:

$$H_{kB}=-\frac{1}{k!}\left(\sum_{\sym{k}}S\right)A^{\oplus k}+\sum_{i=1}^{\ell}U_i R_i$$

Note that this expression allows for a good deal of flexibility in the nature of the interaction between particles.  In particular, it includes the non-interacting case, as well as the Bose-Hubbard model, in which the contribution from each site is proportional to the square of the number of particles at that site.

We are now ready to prove the following lemma:

\begin{lemma} \label{HtoHk} If $\phi$ is a k-equivalence from $G$ to $G\pr$ with corresponding $k$-Boson Hamiltonians $H_{kB}$ and $H\pr_{kB}$, then
\begin{enumerate}[(i)]
\item $\forall S\in \sym{k},\;\widehat\phi(S)=S$
\item $\widehat\phi(A^{\oplus k})=(A\pr)^{\oplus k}$
\item $\forall i,\;\widehat\phi(R_i)=R_i$
\end{enumerate}
and therefore $\widehat{\phi}(H)=H\pr$.
\end{lemma}

\begin{proof} First, we note that any $S\in\sym{k}$ is simply a cylindric relation, with $R_{i,j}\in\{I,J\}$ for all $(i,j)$.  Since $\widehat\phi(I)=I$ and $\widehat\phi(J)=J$, we can apply Lemma \ref{cylindric2} to prove (i).

Similarly, each $R_i$ is a cylindric relation with $R_{i,j}\in \{I,(J-I)\}$, so the same reasoning implies that (iii).

Finally, the definition of $k$-equivalence requires that $\widehat\phi\mid_{W^k}=\phi^k$ where $\phi:W\to W\pr$ is a weak isomorphism such that $\phi(A)=A\pr$.  Therefore, 
\begin{flalign*}
\widehat\phi(A\otimes....\otimes I)=&(A\pr\otimes....\otimes I),\\
\widehat\phi(I\otimes A\otimes....\otimes I)=&(I\otimes A\pr\otimes....\otimes I),\\
&\vdots\\
\widehat\phi(I\otimes...\otimes A)=&(I\otimes...\otimes A\pr)
\end{flalign*}
and $\widehat\phi(A^{\oplus k})=(A\pr)^{\oplus k}$, proving (ii).

Combining (i), (ii) and (iii) gives us $\widehat{\phi}(H)=H\pr$.
\end{proof}

Applying Theorem \ref{indistinguishable} gives us the following generalization of Theorem \ref{2bos}:
\begin{theorem}\label{kequivkbos}If $G$ and $G\pr$ are $k$-equivalent graphs, then they are not distinguished by the interacting $k$-Boson quantum walk.\end{theorem}

\section{Conclusion}

Theorem \ref{kequivkbos}, along with the $k$-equivalent graph constructions given in \cite{Barghi:2009p42}, prove that there is no integer $k$ such that the interacting $k$-Boson quantum walk distinguishes all pairs of non-isomorphic graphs.  This gives a negative answer to a question posed in \cite{Gamble:2010p1085}.  However, while \cite{Gamble:2010p1085} is primarily concerned with strongly regular graphs, the constructions provided by \cite{Barghi:2009p42} are not strongly regular.  The power of multi-particle quantum walks to distinguish strongly regular graphs remains an open question--- a particularly interesting one since strongly regular graphs have particularly simple cellular closures.

\bibliographystyle{plain}
\bibliography{bosons}

\end{document}